\documentclass[9pt,shortpaper,twoside,web]{ieeecolor}
\usepackage{generic}
\usepackage{cite}
\usepackage{amsmath,amssymb,amsfonts}
\usepackage{algorithmic}
\usepackage{graphicx}
\usepackage{textcomp}
\usepackage{subfig}

\def\BibTeX{{\rm B\kern-.05em{\sc i\kern-.025em b}\kern-.08em
    T\kern-.1667em\lower.7ex\hbox{E}\kern-.125emX}}
\begin{document}
\title{Gaussian Process Latent Force Models for Learning and Stochastic Control of Physical Systems}
\author{Simo~S\"arkk\"a,~\IEEEmembership{Senior Member,~IEEE,}
        Mauricio~A.~\'Alvarez, 
        and~Neil~D.~Lawrence%
\thanks{Manuscript received \today; revised XXX.}%
\thanks{Simo~S\"arkk\"a is with the Department of Electrical Engineering and Automation (EEA), Aalto University, Rakentajanaukio 2c, 02150 Espoo, Finland (simo.sarkka@aalto.fi). Tel. +358 50 512 4393}%
\thanks{Mauricio A. \'Alvarez is with the Department of Computer Science, University of Sheffield,
Sheffield, UK S1 4DP}%
\thanks{Neil D. Lawrence is with the Department of Computer Science, University of Sheffield,
Sheffield, UK S1 4DP and with Amazon, Cambridge, UK.}}

\newtheorem{theorem}{Theorem}[section]
\newtheorem{corollary}{Corollary}[section]
\newtheorem{remark}{Remark}[section]
\newtheorem{lemma}{Lemma}[section]

\newcommand{\mauricio}[1]{{\color{blue}[#1]}}
\newcommand{\simo}[1]{{\color{red}[#1]}}

\maketitle

\begin{abstract}
This article is concerned with learning and stochastic control in physical systems which contain unknown input signals. These unknown signals are modeled as Gaussian processes (GP) with certain parametrized covariance structures. The resulting latent force models (LFMs) can be seen as hybrid models that contain a first-principles physical model part and a non-parametric GP model part. We briefly review the statistical inference and learning methods for this kind of models, introduce stochastic control methodology for the models, and provide new theoretical observability and controllability results for them. 
\end{abstract}


\section{Introduction}
\label{sec:introduction}

This article is concerned with the methodology and theory for learning and stochastic control in Gaussian process latent force models (LFMs) \cite{Alvarez+Luengo+Lawrence:2009,Alvarez:2010,Alvarez+Luengo+Lawrence:2013,Hartikainen+Sarkka:2011,Hartikainen+Seppanen+Sarkka:2012}. An example of such LFM is a second order differential equation model of a physical system 
\begin{equation}
\begin{split}
  \frac{\mathrm{d}^{2}f(t)}{\mathrm{d}t^{2}} + \lambda \, \frac{\mathrm{d}f(t)}{\mathrm{d}t}
  + \gamma \, f(t) = u(t) + c(t),
\end{split}
\label{eq:sde0}
\end{equation}
where $\lambda, \gamma > 0$ are parameters of the physical system, the input signal $u(t)$ and the solution $f(t)$ are unknown, and $c(t)$ is a control function to be optimized. Further assume that we measure the function $f(t)$ via noisy measurements at discrete instants of time $t_1,t_2,\ldots,t_n$ via the model  $y_k = f(t_k) + \epsilon_k$, where $\epsilon_k$ is a Gaussian measurement noise and $k = 1,\ldots,n$.

Another example of a problem of interest is the controlled heat equation which we again measure via noisy measurements:
\begin{equation}
\begin{split}
  \frac{\partial f(\mathbf{x},t)}{\partial t} &=
  D \, \nabla^2 \, f(\mathbf{x},t) - \lambda \, f(\mathbf{x},t)
  + u(\mathbf{x},t) + c(\mathbf{x},t),
\end{split}
\label{eq:spde0}
\end{equation}
where $D,\lambda > 0$ are given constants. The aim is to learn both the input signal $u(\mathbf{x},t)$ and the function $f(\mathbf{x},t)$ from noisy observations $y_k = f(\mathbf{x}_k,t_k) + \epsilon_k$, and to design a control $c(\mathbf{x},t)$ for regulating the heat. 

The model \eqref{eq:sde0} is a special case of state-space models of the form
\begin{equation}
\begin{split}
  \frac{\mathrm{d}\mathbf{f}(t)}{\mathrm{d}t}
  &= \mathbf{A}_f \, \mathbf{f}(t) + \mathbf{B}_f \, \mathbf{u}(t) + \mathbf{M}_f \, \mathbf{c}(t), \\
  \mathbf{y}_k &= \mathbf{C}_f \, \mathbf{f}(t) + \boldsymbol{\epsilon}_k.
\end{split}
\label{eq:ssgen}
\end{equation}
where $\mathbf{f}(t)$, $\mathbf{u}(t)$, and $\mathbf{c}(t)$ are vector-valued functions, and $\mathbf{A}_f$, $\mathbf{B}_f$, $\mathbf{C}_f$, and $\mathbf{M}_f$ are given matrices with appropriate dimensions. The second model \eqref{eq:spde0} is a special case of spatio-temporal state-space models
\begin{equation}
\begin{split}
  \frac{\partial \mathbf{f}(\mathbf{x},t)}{\partial t}
  &= \mathbf{\mathcal{A}}_f \, \mathbf{f}(\mathbf{x},t) + \mathbf{B}_f \, \mathbf{u}(\mathbf{x},t) + \mathbf{M}_f \, \mathbf{c}(\mathbf{x},t), \\
  \mathbf{y}_k &= \mathbf{C}_f \, \mathbf{f}(\mathbf{x}_k,t_k) + \boldsymbol{\epsilon}_k,
\end{split}
\label{eq:stsgen}
\end{equation}
where now $\mathbf{\mathcal{A}}_f$ is a matrix of spatial operators and $\mathbf{B}_f$, $\mathbf{C}_f$, and $\mathbf{M}_f$ are given matrices. 

In this article, we specifically concentrate on the above two general classes of models. The aim is to consider the problems of learning (estimating) the functions $\mathbf{f}(\cdot)$ and $\mathbf{u}(\cdot)$ from a set of noisy measurements $\{ \mathbf{y}_k \}$ as well as jointly design the optimal control function $\mathbf{c}(\cdot)$. When the input function $\mathbf{u}(\cdot)$ is modeled as a Gaussian process \cite{Rasmussen+Williams:2006} with a covariance structure allowing for a state-space representation \cite{Hartikainen+Sarkka:2010,Sarkka+Solin+Hartikainen:2013,Sarkka+Piche:2014}, then the models have a tight connection to classical stochastic control theory. In that case it turns out that we can readily apply some of the theory and methodology of Kalman filters and linear quadratic controllers on them provided that we recast the model as an augmented white-noise driven state-space system. 

The main contributions of the article are the stochastic optimal control methods for LFMs as well as the theoretical results in observability and controllability of the models. In particular, we show that although LFMs are observable in quite general conditions, they are never controllable. However, as we discuss in the article, the non-controllability is not a problem in applications, because they still are output-controllable with respect to the physical system part and hence the only uncontrollable part is the unknown input.

The learning methods for LFMs have previously been presented in conference and journal articles \cite{Alvarez+Luengo+Lawrence:2009,Alvarez:2010,Alvarez+Luengo+Lawrence:2013,Hartikainen+Sarkka:2011,Hartikainen+Seppanen+Sarkka:2012}, and they are also closely related to the regularization network methodology considered already earlier in \cite{DeNicolao:2001,DeNicolao:2003}. The learning problem is also related to so called input estimation problem that has previously been addressed in the target tracking literature (e.g. \cite{Bar-Shalom+Li+Kirubarajan:2001}) by replacing the input with a white or colored noise. Another approach to this problem is to use disturbance observers \cite{Chen:2016}. However, here we will specifically concentrate on the Gaussian process based machine learning point of view which allows for encoding prior information into the driving input as well as the use of modern machine learning methods for coping with the related hyperparameter estimation problems and model extensions. 

\subsection{Learning in Gaussian process latent force models}

In machine learning, Gaussian processes (GPs) \cite{Rasmussen+Williams:2006} are commonly used as prior distributions over functions $\mathbf{f}(\boldsymbol{\xi})$. When used for regression, the GP encodes the uncertainty we have over a function, before seeing the data. Given a set of noisy measurements pairs $\mathcal{D} = \{ (\boldsymbol{\xi}_k,\mathbf{y}_k) \}_{k=1}^n$ with, for example, $\mathbf{y}_k = \mathbf{f}(\boldsymbol{\xi}_k) + \boldsymbol{\epsilon}_k$, where $\boldsymbol{\epsilon}_k$ is a vector of Gaussian noises, we can then compute the posterior Gaussian process using the Gaussian process regression equations \cite{Rasmussen+Williams:2006} and use it to make predictions on test points. In the current article we consider cases where $\boldsymbol{\xi} = t$ is the time and $\boldsymbol{\xi} = (\mathbf{x},t)$, where the input consists of both spatial and time components.

In Gaussian process regression notation \cite{Rasmussen+Williams:2006} we write
\begin{equation}
\begin{split}
  \mathbf{f}(\boldsymbol{\xi}) &\sim \mathcal{GP}(\mathbf{0},\mathbf{K}(\boldsymbol{\xi},\boldsymbol{\xi}')), \\
  \mathbf{y}_k &= \mathbf{f}(\boldsymbol{\xi}) + \boldsymbol{\epsilon}_k,
\end{split}
\end{equation}
where $\mathbf{K}(\boldsymbol{\xi},\boldsymbol{\xi}')$ is a given covariance function, and the computational aim is to do inference on the posterior distribution of $f(\cdot)$ conditioned on the measurements $\mathcal{D}$ (obtained by the Bayes' rule) as well as on the parameters of the covariance function. Above, we have, without loss of generality, assumed that the a priori Gaussian process has zero mean.

As shown in \cite{Alvarez+Luengo+Lawrence:2009,Alvarez:2010,Alvarez+Luengo+Lawrence:2013,Hartikainen+Sarkka:2011,Hartikainen+Seppanen+Sarkka:2012}, given a model of the form \eqref{eq:ssgen} with $\mathbf{u}(t) \sim \mathcal{GP}(0,\mathbf{K}(t,t'))$ or a model of the form \eqref{eq:stsgen} with $\mathbf{u}(\mathbf{x},t) \sim \mathcal{GP}(0,\mathbf{K}(\mathbf{x},t;\mathbf{x}',t'))$ the functions $\mathbf{f}(t)$ and $\mathbf{f}(\mathbf{x},t)$ are Gaussian processes as well, and their covariance functions can be expressed in terms of the impulse response or Green's function of the (partial) differential equation together with the covariance function of $\mathbf{u}$. This allows us to reduce inference on LFMs to ordinary GP regression. 

Another point of view is discussed in \cite{Hartikainen+Sarkka:2011,Hartikainen+Seppanen+Sarkka:2012} (see also \cite{DeNicolao:2001,DeNicolao:2003}). In that approach the input GP $\mathbf{u}(t) \sim \mathcal{GP}(0,\mathbf{K}(t,t'))$ is converted into an equivalent state-space representation by using a spectral factorization:
\begin{equation}
\begin{split}
  \frac{\mathrm{d}\mathbf{z}(t)}{\mathrm{d}t}
  &= \mathbf{A}_u \, \mathbf{z}(t) + \mathbf{B}_u \, \mathbf{w}(t), \\
  \mathbf{u}(t) &= \mathbf{C}_u \, \mathbf{z}(t).
\end{split}
\label{eq:ssu}
\end{equation}
Here the state-vector typically consists of a set of derivatives of the process $\mathbf{z} = (\mathbf{u},\mathrm{d}\mathbf{u}/\mathrm{d}t,\ldots,\mathrm{d}^{s-1}\mathbf{u}/\mathrm{d}t^{s-1})$, and $\mathbf{w}(t)$ is a vector-valued white-noise process with a given spectral density matrix. The advantage of this kind of model formulation is that it allows for solving the GP regression problem using Kalman filters and smoothers \cite{Sarkka:2013} in $O(n)$ time when the traditional GP takes $O(n^3)$ time (here $n$ denotes the number of measurements). 

The same idea can be extended to spatio-temporal Gaussian processes \cite{Sarkka+Hartikainen:2012,Sarkka+Solin+Hartikainen:2013}. The conversion of a spatio-temporal covariance function into state-space form leads to a system of the form
\begin{equation}
\begin{split}
  \frac{\partial\mathbf{z}(\mathbf{x},t)}{\partial t} &= \mathbf{\mathcal{A}}_u \, \mathbf{z}(\mathbf{x},t)
  + \mathbf{B}_u \, \mathbf{w}(\mathbf{x},t), \\
   \mathbf{u}(\mathbf{x},t) &= \mathbf{C}_u \, \mathbf{z}(\mathbf{x},t),
\end{split}
\label{eq:ssu2}
\end{equation}
where $\mathbf{\mathcal{A}}_u$ is a matrix of linear operators (typically pseudo-differential operators) acting on the $\mathbf{x}$-variable and $\mathbf{w}(\mathbf{x},t)$ is a vector-valued time-white spatio-temporal Gaussian process with a given spectral density kernel. In this case the inference can be done using infinite-dimensional Kalman filters and smoothers which typically are approximated with their finite-dimensional counterparts. More details can be found in \cite{Sarkka+Hartikainen:2012,Sarkka+Solin+Hartikainen:2013}.

We can now combine the  state-space ODE \eqref{eq:ssgen} with the state-space representation of LFMs to obtain an augmented state-space representation of the LFM \cite{Hartikainen+Sarkka:2011,Hartikainen+Seppanen+Sarkka:2012}:
\begin{equation}
\begin{split}
  \frac{\mathrm{d}\mathbf{f}(t)}{\mathrm{d}t}
  &= \mathbf{A}_f \, \mathbf{f}(t)
  + \mathbf{B}_f \, \mathbf{C}_u \, \mathbf{u}(t) + \mathbf{M}_f \, \mathbf{c}(t), \\
  \frac{\mathrm{d}\mathbf{u}(t)}{\mathrm{d}t}
  &= \mathbf{A}_u \, \mathbf{u}(t) + \mathbf{B}_u \, \mathbf{w}(t), \\
  \mathbf{y}_k &= \mathbf{C}_f \, \mathbf{f}(t) + \boldsymbol{\epsilon}_k.
\end{split}
\label{eq:comb}
\end{equation}
If we now define
\begin{equation}
\begin{split}
  \mathbf{g} &= \begin{pmatrix}
	\mathbf{f} \\ \mathbf{u}
  \end{pmatrix}, \quad
  \mathbf{A}
  = \begin{pmatrix}
	\mathbf{A}_f & \mathbf{\mathbf{B}_f \, \mathbf{C}_u} \\
	\mathbf{0} & \mathbf{A}_u
  \end{pmatrix}, \quad
  \mathbf{M} = \begin{pmatrix} \mathbf{M}_f \\ \mathbf{0} \end{pmatrix}\\
  \mathbf{B}
  &= \begin{pmatrix}
	\mathbf{0} & \mathbf{B}_u
  \end{pmatrix}, \quad
  \mathbf{C}
  = \begin{pmatrix}
	\mathbf{C}_f & \mathbf{0}
  \end{pmatrix},
\end{split}
\label{eq:ssaugmats}
\end{equation}
then the model can be written as a white-noise driven model
\begin{equation}
\begin{split}
  \frac{\mathrm{d}\mathbf{g}(t)}{\mathrm{d}t}
  &= \mathbf{A} \, \mathbf{g}(t)
  + \mathbf{B} \, \mathbf{w}(t) + \mathbf{M} \, \mathbf{c}(t), \\
  \mathbf{y}_k &= \mathbf{C} \, \mathbf{g}(t) + \boldsymbol{\epsilon}_k.
\end{split}
\label{eq:ssaug}
\end{equation}
Spatio-temporal models \eqref{eq:stsgen} driven by Gaussian processes can also be often represented in a similar state-space form, which now becomes
\begin{equation}
\begin{split}
  \frac{\partial \mathbf{f}(\mathbf{x},t)}{\partial t}
  &= \mathbf{\mathcal{A}}_f \, \mathbf{f}(\mathbf{x},t)
  + \mathbf{B}_f \, \mathbf{C}_u \, \mathbf{z}(\mathbf{x},t) + \mathbf{M}_f \, \mathbf{c}(\mathbf{x},t), \\
  \frac{\partial\mathbf{z}(\mathbf{x},t)}{\partial t} &= \mathbf{\mathcal{A}}_u \, \mathbf{z}(\mathbf{x},t)
  + \mathbf{B}_u \, \mathbf{w}(\mathbf{x},t), \\
  \mathbf{y}_k &= \mathbf{C}_f \, \mathbf{f}(\mathbf{x}_k,t_k) + \boldsymbol{\epsilon}_k.
\end{split}
\end{equation}
In order to obtain a single augmented model, we can define
\begin{equation}
\begin{split}
  \mathbf{g} &= \begin{pmatrix}
	\mathbf{f} \\ \mathbf{u}
  \end{pmatrix}, \quad
  \mathbf{\mathcal{A}}
  = \begin{pmatrix}
	\mathbf{\mathcal{A}}_f & \mathbf{B}_f \, \mathbf{C}_u \\
	\mathbf{0} & \mathbf{\mathcal{A}}_u
  \end{pmatrix}, \quad
  \mathbf{M} = \begin{pmatrix} \mathbf{M}_f \\ \mathbf{0} \end{pmatrix} \\
  \mathbf{B}
  &= \begin{pmatrix}
	\mathbf{0} & \mathbf{B}_u
  \end{pmatrix}, \quad
  \mathbf{C}
  = \begin{pmatrix}
	\mathbf{C}_f & \mathbf{0}
  \end{pmatrix},
\end{split}
\end{equation}
which leads to a model of the form
\begin{equation}
\begin{split}
  \frac{\partial \mathbf{g}(\mathbf{x},t)}{\partial t}
  &= \mathbf{\mathcal{A}} \, \mathbf{g}(\mathbf{x},t)
  + \mathbf{B} \, \mathbf{w}(\mathbf{x},t) + \mathbf{M} \, \mathbf{c}(t), \\
  \mathbf{y}_k &= \mathbf{C} \, \mathbf{g}(\mathbf{x}_k,t_k) + \boldsymbol{\epsilon}_k.
\end{split}
\label{eq:ssaug2}
\end{equation}
The joint state-space representations \eqref{eq:ssaug} and \eqref{eq:ssaug2} of the LFMs now allows for full Bayesian inference in the models to be performed with Kalman filtering and smoothing methods \cite{Hartikainen+Sarkka:2011,Hartikainen+Seppanen+Sarkka:2012,Sarkka+Solin+Hartikainen:2013}. Furthermore, these representations also allow us to study control problems on LFMs which aim at designing controller functions $\mathbf{c}$. This problem is addressed in the next section.

\section{Stochastic Control of Gaussian Process Latent Force Models}
\label{sec:lfmcontrol}

In this section, we discuss the stochastic control problems related to latent force models. In particular, we provide and analyze the solutions for the linear quadratic regulation (LQR) problem for them.

\subsection{Controlled temporal LFMs} \label{sec:controlled_basic}
Let us consider the state-space model with a Gaussian process input \eqref{eq:ssgen}:
\begin{equation}
\begin{split}
  \frac{\mathrm{d}\mathbf{f}(t)}{\mathrm{d}t}
  &= \mathbf{A}_f \, \mathbf{f}(t) + \mathbf{B}_f \, \mathbf{u}(t) + \mathbf{M}_f \, \mathbf{c}(t).
\end{split}
\end{equation}
We will specifically aim to consider optimal control problems which minimize the quadratic cost functional
\begin{equation}
\begin{split}
  \mathcal{J}[\mathbf{c}] &= \frac{1}{2} \mathrm{E} \Big[
    \mathbf{f}^{\top}(T) \, \boldsymbol{\Phi} \, \mathbf{f}(T) \\
   &+ \int_0^T
   (\mathbf{f}^{\top}(t) \, \mathbf{X}(t) \, \mathbf{f}(t)
  + \mathbf{c}^{\top}(t) \, \mathbf{U}(t) \, \mathbf{c}(t)) \, \mathrm{d}t \Big],
\end{split}
\label{eq:quadcost1}
\end{equation}
where $\mathrm{E}[\cdot]$ denotes the expected value, $\boldsymbol{\Phi}$, $\mathbf{X}(t)$, and $\mathbf{U}(t)$ are positive semidefinite matrices for all $t \ge 0$, and $T$ is the target time, because they lead to computationally tractable control laws. However, the principle outlined here can also be extended to more general cost functionals although the numerical methods become order of magnitude more complicated.

A straightforward approach to optimal control with the quadratic cost \eqref{eq:quadcost1} is to use the separation principle of linear estimation and control which amounts to designing the optimal controller for the case $\mathbf{u}(t) = \mathbf{0}$ and use it in cascade with a Kalman filter. This indeed is the optimal solution in the case of white $\mathbf{u}(t)$, but not in our case.

The correct approach in this case, which also utilizes the learning outcome of the Gaussian process regression is to use the augmented state space model with the control signal. In this case it is given as (see \eqref{eq:ssaug})

\begin{equation}
\begin{split}
  \frac{\mathrm{d}\mathbf{g}(t)}{\mathrm{d}t}
  &= \mathbf{A} \, \mathbf{g}(t)
  + \mathbf{B} \, \mathbf{w}(t) + \mathbf{M} \, \mathbf{c}(t),
\end{split}
\end{equation}
with the measurement model given in \eqref{eq:ssaug} and the matrices $\mathbf{A}$ and $\mathbf{B}$ as defined in \eqref{eq:ssaugmats}. We now aim to design a controller for the above model by assuming a perfectly observed state and run it in cascade with a Kalman filter processing the measurements in the model. This yields to a controller which jointly learns the functions $\mathbf{f}$ and $\mathbf{u}$ and jointly optimizes the control with respect to the cost criterion \cite{Maybeck:1982b}. In this case the control cost function can be rewritten in form
\begin{equation}
\begin{split}
  \mathcal{J}[\mathbf{c}] &= \frac{1}{2} \mathrm{E} \Big[
    \mathbf{g}^{\top}(T) \, \boldsymbol{\Phi}_g \, \mathbf{g}(T) \\
   &+ \int_0^T
   (\mathbf{g}^{\top}(t) \, \mathbf{X}_g(t) \, \mathbf{g}(t)
  + \mathbf{c}^{\top}(t) \, \mathbf{U}(t) \, \mathbf{c}(t)) \, \mathrm{d}t \Big],
\end{split}
\end{equation}
where
\begin{equation}
\begin{split}
 \boldsymbol{\Phi}_g &= \begin{pmatrix}
   \boldsymbol{\Phi} & \mathbf{0} \\ \mathbf{0} & \mathbf{0}
 \end{pmatrix}, \qquad
 \mathbf{X}_g(t) = \begin{pmatrix}
   \mathbf{X}(t) & \mathbf{0} \\ \mathbf{0} & \mathbf{0}
 \end{pmatrix}.
\end{split}
\end{equation}
The design of the optimal linear quadratic controller for the resulting model can be done by using the classical Riccati-equation-based approaches \cite{Kalman:1960b,Anderson+Moore:2007}. Namely, the optimal control takes the form
\begin{equation}
\begin{split}
  \mathbf{c}(t) = -\mathbf{U}^{-1}(t) \, \mathbf{M}^{\top} \, \mathbf{P}(t) \, \hat{\mathbf{g}}(t),
\end{split}
\label{eq:joint_control1} 
\end{equation}
where $\hat{\mathbf{g}}(t)$ is the Kalman filter estimate of $\mathbf{g}(t)$ and the matrix $\mathbf{P}(t)$ solves the backward Riccati differential equation
\begin{equation}
\begin{split}
  \frac{\mathrm{d}\mathbf{P}(t)}{\mathrm{d}t} &=
    -\mathbf{A}^{\top} \, \mathbf{P}(t) - \mathbf{P}(t) \, \mathbf{A} \\
  &+  \mathbf{P}(t) \, \mathbf{M} \, \mathbf{U}^{-1}(t) \,
   \mathbf{M}^{\top} \, \mathbf{P}(t) - \mathbf{X}_g(t)
\end{split}
\label{eq:bgricc}
\end{equation}
with the boundary condition $\mathbf{P}(T) =  \boldsymbol{\Phi}_g$. However, we can write this solution for the LFM model in more explicit form which reveals its structure better. That is summarized in the following theorem. 

\begin{theorem} \label{the:control1} 
The control law in \eqref{eq:joint_control1} can be written as
\begin{equation} 
\begin{split}
  \mathbf{c}(t) = - \begin{pmatrix}
    \mathbf{U}^{-1} \, \mathbf{M}_f^{\top} \, \mathbf{P}_{f}(t) &
    \mathbf{U}^{-1} \, \mathbf{M}_f^{\top} \, \mathbf{P}_{12}(t)
  \end{pmatrix} \, \hat{\mathbf{g}}(t),
\end{split}
\end{equation}
where $\mathbf{P}_{f}(t) \triangleq \mathbf{P}_{11}(t)$ is the Riccati equation solution for the non-forced physical model. The full set of equations is
\begin{equation}
\begin{split}
  \frac{\mathrm{d}\mathbf{P}_{11}(t)}{\mathrm{d}t}
  &= -\mathbf{A}_f^{\top} \, \mathbf{P}_{11} - \mathbf{P}_{11} \, \mathbf{A}_f
  \\
  &\qquad
  + \mathbf{P}_{11} \, \mathbf{M}_f \, \mathbf{U}^{-1} \, \mathbf{M}_f^{\top} \, \mathbf{P}_{11} - \mathbf{X}(t), \\
  \frac{\mathrm{d}\mathbf{P}_{12}(t)}{\mathrm{d}t}
  &= -\mathbf{A}_f^{\top} \, \mathbf{P}_{12} - \mathbf{P}_{11} \, \mathbf{B}_f \, \mathbf{C}_u - \mathbf{P}_{12} \, \mathbf{A}_u
  \\
  &\qquad + \mathbf{P}_{11} \, \mathbf{M}_f \, \mathbf{U}^{-1} \, \mathbf{M}_f^{\top} \, \mathbf{P}_{12}, \\
  \frac{\mathrm{d}\mathbf{P}_{22}(t)}{\mathrm{d}t}
  &= -\mathbf{C}_u^{\top} \, \mathbf{B}_f^{\top} \, \mathbf{P}_{12} - \mathbf{A}_u^{\top} \, \mathbf{P}_{22}
  - \mathbf{P}_{21} \, \mathbf{B}_f \, \mathbf{C}_u \\
  &\qquad  - \mathbf{P}_{22} \, \mathbf{A}_u
  + \mathbf{P}_{12}^{\top} \, \mathbf{M}_f \, \mathbf{U}^{-1} \, \mathbf{M}_f^{\top} \, \mathbf{P}_{12}.
\end{split}
\end{equation}
\end{theorem}
\begin{proof}
The result can be obtained by inserting the partitioned $\mathbf{P} = \begin{pmatrix} \mathbf{P}_{11} & \mathbf{P}_{12} \\ \mathbf{P}_{12}^{\top} & \mathbf{P}_{22} \end{pmatrix}$ into \eqref{eq:bgricc}.
\end{proof}
In the above theorem the gain for the physical system (i.e. $\mathbf{f}$) portion of the state is exactly the same as in the optimal controller without an input. However, the second part of gain is non-zero and uses the input states for control feedback as well.

In the next section we will simplify the control problem even more, and consider the limit $T \to \infty$, because it leads to a particularly convenient class of linear controllers which are computationally tractable while still being able to use the learning outcome of the Gaussian process inference.

\subsection{Linear quadratic regulation of temporal LFMs}
In the LFM case, namely because we have restricted our consideration to time-invariant models, a very convenient type of control problem is the infinite-time linear regulation problem which corresponds to the cost function
\begin{equation}
\begin{split}
  \mathcal{J}[\mathbf{c}] &= \int_0^\infty
   (\mathbf{f}^{\top}(t) \, \mathbf{X} \, \mathbf{f}(t)
  + \mathbf{c}^{\top}(t) \, \mathbf{U} \, \mathbf{c}(t)) \, \mathrm{d}t \Big],
\end{split}
\end{equation}
where $\mathbf{X}$ and $\mathbf{U}$ are constant semidefinite matrices. By rewriting the model as an augmented state-space model as we did in the previous section and by following the classical results, the controller becomes
\begin{equation}
\begin{split}
  \mathbf{c}(t) = - \mathbf{U}^{-1} \, \mathbf{M}^{\top} \, \mathbf{P}\, \hat{\mathbf{g}}(t),
\end{split}
\label{eq:joint_control}
\end{equation}
where the matrix $\mathbf{P}$ is the solution to the algebraic Riccati equation (ARE) 
\begin{equation}
\begin{split}
  \mathbf{0} &=
    -\mathbf{A}^{\top} \, \mathbf{P} - \mathbf{P} \, \mathbf{A} 
  +  \mathbf{P} \, \mathbf{M} \, \mathbf{U}^{-1} \,
   \mathbf{M}^{\top} \, \mathbf{P} - \mathbf{X}_g,
\end{split}
\label{eq:joint_are}
\end{equation}
where 
 $\mathbf{X}_g = \begin{pmatrix}
   \mathbf{X} & \mathbf{0} \\ \mathbf{0} & \mathbf{0}
 \end{pmatrix}$.

By solving the control law from these equations, we get a controller which is function of both the estimate of the function $\mathbf{f}$ and estimate of the input $\mathbf{u}$. Thus this control law is able to utilize both the estimate of the function as well the learned input function.

It is also possible to express the solution to the LFM control problem above in terms of the corresponding control solution to the non-forced problem similarly to the time-varying case considered in the previous section. This result is summarized in the following theorem. 

\begin{theorem} \label{the:control} 
The control law in \eqref{eq:joint_control} can now be written as
\begin{equation} 
\begin{split}
  \mathbf{c}(t) = - \begin{pmatrix}
    \mathbf{U}^{-1} \, \mathbf{M}_f^{\top} \, \mathbf{P}_{f} &
    \mathbf{U}^{-1} \, \mathbf{M}_f^{\top} \, \mathbf{P}_{12}
  \end{pmatrix} \, \hat{\mathbf{g}}(t),
\end{split}
\end{equation}
where $\mathbf{U}^{-1} \, \mathbf{M}_f^{\top} \, \mathbf{P}_{f}$ is just the non-forced-case gain and $\mathbf{P}_{12}$ can be solved from the Sylvester equation
\begin{equation}
\begin{split}
  \left( \mathbf{P}_{f} \, \mathbf{M}_f \, \mathbf{U}^{-1} \, \mathbf{M}_f^{\top} 
         - \mathbf{A}_f^{\top} \right) \, \mathbf{P}_{12}  - \mathbf{P}_{12} \, \mathbf{A}_u
    = \mathbf{P}_{f} \, \mathbf{B}_f \, \mathbf{C}_u.
\end{split}
\end{equation}
\end{theorem}
\begin{proof}
The result can be obtained by setting the time derivatives in Theorem~\ref{the:control1} to zero.\end{proof}

Note that although the system is stabilizable also by setting the second term to zero, that is, using the non-forced gain (cf. Theorem~\ref{the:stab}), a better solution than that is obtained by using the control in Theorem~\ref{the:control} which depends on the input as well.

\subsection{Controlled spatio-temporal LFMs}
In the case of PDE LFMs we get models of the form
\begin{equation}
\begin{split}
  \frac{\partial \mathbf{g}(\mathbf{x},t)}{\partial t}
  &= \mathbf{\mathcal{A}} \, \mathbf{g}(\mathbf{x},t)
  + \mathbf{B} \, \mathbf{w}(\mathbf{x},t) + \mathbf{M}_f \, \mathbf{c}(\mathbf{x},t), \\
  \mathbf{y}_k &= \mathbf{C} \, \mathbf{g}(t_k) + \boldsymbol{\epsilon}_k,
\end{split}
\end{equation}
where the control problem corresponds to designing the control function $\mathbf{c}(\mathbf{x},t)$ minimizing, for example, a linear quadratic cost functional. In principle, it is possible to directly analyze such infinite-dimensional control problems which leads to, for example, generalizations of the controllability concepts \cite{Curtain:2012}. However, in practice, after setting up the model, we replace the infinite-dimensional model with its finite-dimensional approximation. Therefore it is actually more fruitful to directly analyze the finite-dimensional approximation rather than the original infinite-dimensional model---this way we can also easily account for the effect of discretization. For the finite-dimensional approximate model the results in the previous and next sections apply as such.

\section{Observability and Controllability}
\label{sec:lfmtheory}

In this section, our aim is to discuss the detectability and observability of the latent force models along with the stabilizability and controllability of them. We only consider finite-dimensional models, because as discussed above, infinite-dimensional models anyway need to be discretized and in order to ensure the detectability and observability of the resulting models, the finite-dimensional results are sufficient. The corresponding pure infinite-dimensional results could be derived using the results in \cite{Curtain:2012}.

\subsection{Detectability and observability of latent force models}

Let us now consider the detectability and observability of LFMs. We assume that we have a latent force model which has the following state space representation

\begin{equation}
\begin{split}
  \frac{\mathrm{d}\mathbf{g}(t)}{\mathrm{d}t}
  &= \mathbf{A} \, \mathbf{g}(t)
  + \mathbf{B} \, \mathbf{w}(t), \\
  \mathbf{y}_k &= \mathbf{C} \, \mathbf{g}(t_k) + \boldsymbol{\epsilon_k},
\end{split}
\label{eq:detss}
\end{equation}
where $\mathbf{g}$ and the matrices $\mathbf{A}$, $\mathbf{B}$, and $\mathbf{C}$ are defined as in \eqref{eq:ssaugmats}. In this representation we have dropped the control signal, because it does not affect the detectability and observability.

It is also reasonable to assume that the state-space representation of the latent force model is stable and hence detectable. However, the physical system part itself often is not stable. We need to assume though that it is at least detectable and preferably it should be observable. The most useful case occurs when the whole joint system is observable. The sampling procedure also affects the observability and we need to ensure that we do not get 'aliasing' kind of phenomenon analogously to sampling a signal with a sampling frequency that is below the Nyquist frequency. Let us start with the following result for detectability.

\begin{lemma} \label{lem:detect}
  Assume that we have a latent force model which has the state space representation given in \eqref{eq:detss}.
Assume that $(\exp(\mathbf{A}_f \, \Delta t_k),\mathbf{C}_f)$ is detectable, and that the input function $u(t)$ has an exponentially stable state space representation. Then the full system is \emph{detectable} and the Kalman filter for the model is exponentially stable.
\end{lemma}

\begin{proof}
We first discretize the system at arbitrary time points. The discretized system has the form (see, e.g., \cite{Bar-Shalom+Li+Kirubarajan:2001})
\begin{equation}
\begin{split}
   \mathbf{g}_k &= \exp(\mathbf{A}_f \, \Delta t_k) \, \mathbf{g}_{k-1} +  \mathbf{q}_k, \\
  \mathbf{y}_k &=  \mathbf{C} \,  \mathbf{g}_k + \boldsymbol{\epsilon_k},
\end{split}
\end{equation}
where $\mathbf{q}_k$ is a Gaussian random variable, which will be detectable provided that there exists a bounded gain sequence $\mathbf{G}_k$ such that the sequence $\tilde{\mathbf{g}}_k$ defined as $\tilde{\mathbf{g}}_k = (\exp(\mathbf{A}_f \, \Delta t_k) - \mathbf{G}_k \, \mathbf{C}) \, \tilde{\mathbf{g}}_{k-1}$ is exponentially stable \cite{Anderson:1981}. More explicitly, the following system for the sequences $\tilde{\mathbf{f}}_k$ and $\tilde{\mathbf{u}}_k$ needs to be exponentially stable with some choice of sequence $\mathbf{G}_k$:
\begin{equation}
\begin{split}
  \tilde{\mathbf{f}}_k &= \exp(\mathbf{A}_f \, \Delta t_k) \, \tilde{\mathbf{f}}_{k-1}
  + \Gamma_k \, \tilde{\mathbf{u}}_{k-1} - \mathbf{G}_k \, \mathbf{C}_f \, \tilde{\mathbf{f}}_{k-1}, \\
  \tilde{\mathbf{u}}_k &= \exp(\mathbf{A}_u \, \Delta t_k) \, \tilde{\mathbf{u}}_{k-1}.
\end{split}
\end{equation}
As the process $\mathbf{u}_k$ is exponentially stable, the sequence $\tilde{\mathbf{u}}_k$ is exponentially decreasing and bounded. Hence it does not affect the stability of the first equation. Therefore, the full system will be detectable provided that there exists a gain sequence $K_k$ such that $\tilde{\mathbf{f}}_k = (\exp(\mathbf{A}_f \, \Delta t_k) - \mathbf{G}_k \, \mathbf{C}_f) \, \tilde{\mathbf{f}}_{k-1}$ is exponentially stable. The gain sequence exists, because $(\exp(\mathbf{A}_f \, \Delta t_k),\mathbf{C}_f)$ is detectable by assumption. 
\end{proof}

Above, in Lemma~\ref{lem:detect} we had to assume the detectability of the discretized system. There are many ways to assure this, but one way is to demand that the continuous physical model is observable and that we are not sampling critically \cite{Ding:2009}, that is, in a way that would lead to aliasing of frequencies as in the Shannon-Nyquist theory. Although observability is a quite strong condition compared to detectability, it assures that we have the chance to reconstruct the physical system with an arbitrary precision by improving the measurement protocol, which would not be true for mere detectability. %

If we assume that the physical system part is observable and the sampling is not critical, we get the following detectability theorem. Note that we do not yet assume that the latent force model part would be observable although its stability already implies that it is detectable.
\begin{theorem} \label{the:detect}
Assume that $(\mathbf{A}_f,\mathbf{C}_f)$ is observable, the physical system is not critically sampled, and that the latent force model part is stable. Then the full system is detectable and the Kalman filter for the model is exponentially stable.
\end{theorem}

\begin{proof}
According to \cite{Ding:2009}, the observability of the continuous-time system together with the non-critical sampling ensures that the discrete-time system is also observable. As discrete-time observability implies discrete-time detectability the result follows from Lemma~\ref{lem:detect}.
\end{proof}

Let us now consider the conditions for the observability of the full system. It turns out that in general, the best way to determine the observability of the joint system is not to attempt to think of the physical system and the latent force model separately, but explicitly consider the joint state-space model. There are numerous attempts to map the properties of this kind cascaded systems to the properties of the joint system (e.g. \cite{Gilbert:1963,Chen:1967,Davison:1975}), but still the best way to go seems to be simply to use a standard observability tests on the joint system. The properties of the sub-systems of this kind of cascade do not alone determine the observability, because we can have phenomena like zero-pole cancellation which leads to a non-observable system even when all the subsystems are observable (see, e.g., \cite{Gilbert:1963}). When we also account for the effect of sampling to observability, we get the following theorem.

\begin{theorem} \label{the:obsv}
Assume that the continuous-time joint system $(\mathbf{A},\mathbf{C})$ is observable, and the
observations are not critically sampled, then the discrete-time full system is observable.
\end{theorem} 

\begin{proof}
See \cite{Ding:2009}.
\end{proof}
In practical terms it is thus easiest to use, for example, the classical rank-condition (see, e.g., \cite{Ogata:1997}) which says that the (joint) system $(\mathbf{A},\mathbf{C})$ is observable, which in time-invariant case is ensured provided that the following matrix has full rank for some $m$:
\begin{equation}
  \mathcal{O} = 
  \begin{pmatrix}
    \mathbf{C} \\
    \mathbf{C} \, \mathbf{A} \\
    \vdots \\
    \mathbf{C} \, \mathbf{A}^{m-1}    
  \end{pmatrix},
\end{equation}
and then ensure that sampling is non-critical \cite{Ding:2009}. Fortunately, the continuous-time joint system will be observable in many practical scenarios provided that we do not have any zero-pole cancellations between the physical system and force model.

\subsection{Stabilizability and non-controllability of LFMs}
The aim is now to discuss the controllability and stabilizability of state-space latent force models. We assume that the model has the form
\begin{equation}
\begin{split}
  \frac{\mathrm{d}\mathbf{g}(t)}{\mathrm{d}t}
  &= \mathbf{A} \, \mathbf{g}(t)
  + \mathbf{B} \, \mathbf{w}(t) + \mathbf{M} \, \mathbf{c}(t),
\end{split}
\end{equation}
where $\mathbf{g}$ and the matrices $\mathbf{A}$, $\mathbf{B}$, and $\mathbf{M}$ are defined in \eqref{eq:ssaugmats}. %

First of all, the stabilizability of the system is guaranteed solely by ensuring that the physical model part is stabilizable, provided that the state-space representation of the stationary GP is constructed such that it is exponentially stable. Thus we have the following theorem.

\begin{theorem} \label{the:stab}
Assume that $(\mathbf{A}_f,\mathbf{M}_f)$ is stabilizable and the latent force has an exponentially stable state space representation. Then the full system is \emph{stabilizable}.
\end{theorem}

\begin{proof}
The system is stabilizable if there exist a finite gain $\mathbf{G}_c$ such that the system $\mathrm{d}\tilde{\mathbf{g}}/\mathrm{d}t = (\mathbf{A} + \mathbf{M} \, \mathbf{G}_c) \, \tilde{\mathbf{g}}$ is exponentially stable \cite{Wonham:1985}. More explicitly we should have
\begin{equation}
\begin{split}
  \frac{\mathrm{d}\tilde{\mathbf{f}}}{\mathrm{d}t} &= (\mathbf{A}_f + \mathbf{M}_f \, \mathbf{G}_f) \, \tilde{\mathbf{f}}
  + (\mathbf{B}_f \, \mathbf{C}_u + \mathbf{M}_f \, \mathbf{G}_u) \, \tilde{\mathbf{u}}, \\
  \frac{\mathrm{d}\tilde{\mathbf{u}}}{\mathrm{d}t} &= \mathbf{A}_u \, \tilde{\mathbf{u}},
\end{split}
\end{equation}
where we have written $\mathbf{G}_c = \begin{pmatrix} \mathbf{G}_f & \mathbf{G}_u \end{pmatrix}$. Because $\tilde{\mathbf{u}}$ is exponentially decreasing and bounded, we can safely set $\mathbf{G}_u = 0$. The remainder of the system will be stabilizable if there exists a gain $\mathbf{G}_f$ such that $\mathrm{d}\tilde{\mathbf{f}}/\mathrm{d}t = (\mathbf{A}_f + \mathbf{M}_f \, \mathbf{G}_f) \, \tilde{\mathbf{f}}$ is exponentially stable. By our assumption on the stabilizability of $(\mathbf{A}_f,\mathbf{M}_f)$, this is true and hence the result follows.
\end{proof}

The stabilizability also implies that the corresponding LQ controller is uniquely determined \cite{Anderson+Moore:2007}. However, the sole stabilizability is not very useful in practice, because sole stabilizability says that we might have randomly wandering subprocesses in the joint system which practically prevent us from controlling the process exactly where we wish it to go. A much stronger requirement is to require that the full system is controllable. Unfortunately, it turns out that latent force models are never fully controllable in the present formulation, because we cannot control the subsystem corresponding to the GP force. This is summarized in the following theorem.

\begin{theorem}
  Latent force models are not controllable.
\end{theorem}

\begin{proof}
  The model is in Kalman's canonical form \cite{Kalman:1963}, where the non-controllable part is the input signal.
\end{proof}

In practice, the non-controllability of the input part is not a problem, as we are actually interested in controlling the physical system part of the model, not the input signal per se. It turns out that the physical system can be controllable even though the full system is not controllable. This result can be obtained as a corollary of so called output controllability (see, e.g., \cite{Ogata:1997}) as follows.

\begin{corollary} %
Assume that $(\mathbf{A}_f,\mathbf{M}_f)$ is controllable. Then the full system is output controllable with respect to the physical system part.
\end{corollary}

\begin{proof}
This can be derived by writing down the output controllability condition \cite{Ogata:1997} and noticing that it reduces to controllability of the physical system part.
\end{proof}

The above result is useful when the system is fully observable as well. Then it ensures that we can successfully control the physical system part although the full latent force model remains uncontrollable. However, if the latent force model is not fully observable, then the latent force model inherently causes disturbance to the physical system and although we can keep the system stable, the state cannot be forced to follow a given trajectory.

As a conclusion, for all practical purposes a (time-invariant) latent force model is controllable, if it is observable and the following matrix has a full rank for some $m$:
\begin{equation}
  \mathcal{C} = \begin{pmatrix}
  	\mathbf{M}_f & \mathbf{A}_f \, \mathbf{M}_f & \hdots & \mathbf{A}_f^{m-1} \, \mathbf{M}_f
  \end{pmatrix}.
\end{equation}

\section{Experimental Results}
\label{sec:experimental}

In this section, we illustrate the latent force model framework in two different problems: a controlled second order ordinary differential equation modeling a spring and a controlled heat source in two dimensions.

\subsection{Controlled ODE Model} \label{sec:ctrl_spring}

Our first illustrative example corresponds to the second order differential equation model described in \eqref{eq:sde0}, which physically can be considered as a damped spring. We consider a 100-second interval, where the first 50 seconds are used for learning the hyperparameters of the (state-space) GP after which the hyperparameters are kept fixed. We then continue obtaining 40 seconds of additional measurements of the system after which the measurements stop while we still continue to run the system for 10 seconds. 

The unknown input signal is $u(t) = \sin(0.23 \, t) + \sin(0.13 \, t)$ for $t \in [0,100]$, the parameters $\lambda = 0.1$ and $\gamma = 1$, and we assume that only the position of the spring $f(t)$ is measured in time intervals of $\Delta t = 0.01$ seconds. The measurements contain Gaussian noise with a relatively small standard deviation $0.01$ -- the small noise is selected to better highlight the differences between the controllers. 

We selected the Gaussian process prior for the input process $u(t)$ to have a zero mean and squared exponential (SE) covariance function of the form $K(t,t') = \sigma^2 \exp[-(t-t')^2 / \ell^2]$ which was approximated with state-space model using 4/8-order Pad\'e approximant \cite{Sarkka+Piche:2014}. During the training phase, the parameters $\sigma$ and $\ell$ were estimated by maximizing the marginal likelihood. The simulated open-loop system along with the Gaussian process interpolation (implemented in state-space with a Rauch--Tung--Striebel smoother) and extrapolation results are shown in Figure~\ref{ode2:output:subfig:mlgp}. It can be seen that the GP follows the true position well until the end of the measurements, after which it quite quickly reverts to the prior mean (which in this case is zero). Thus the extrapolation accuracy of the GP model is fairly limited, but fortunately the uncertainty estimate of the GP indicates that this should be expected. 

\begin{figure}[!ht]
\centering
       \includegraphics[width=0.9\columnwidth]{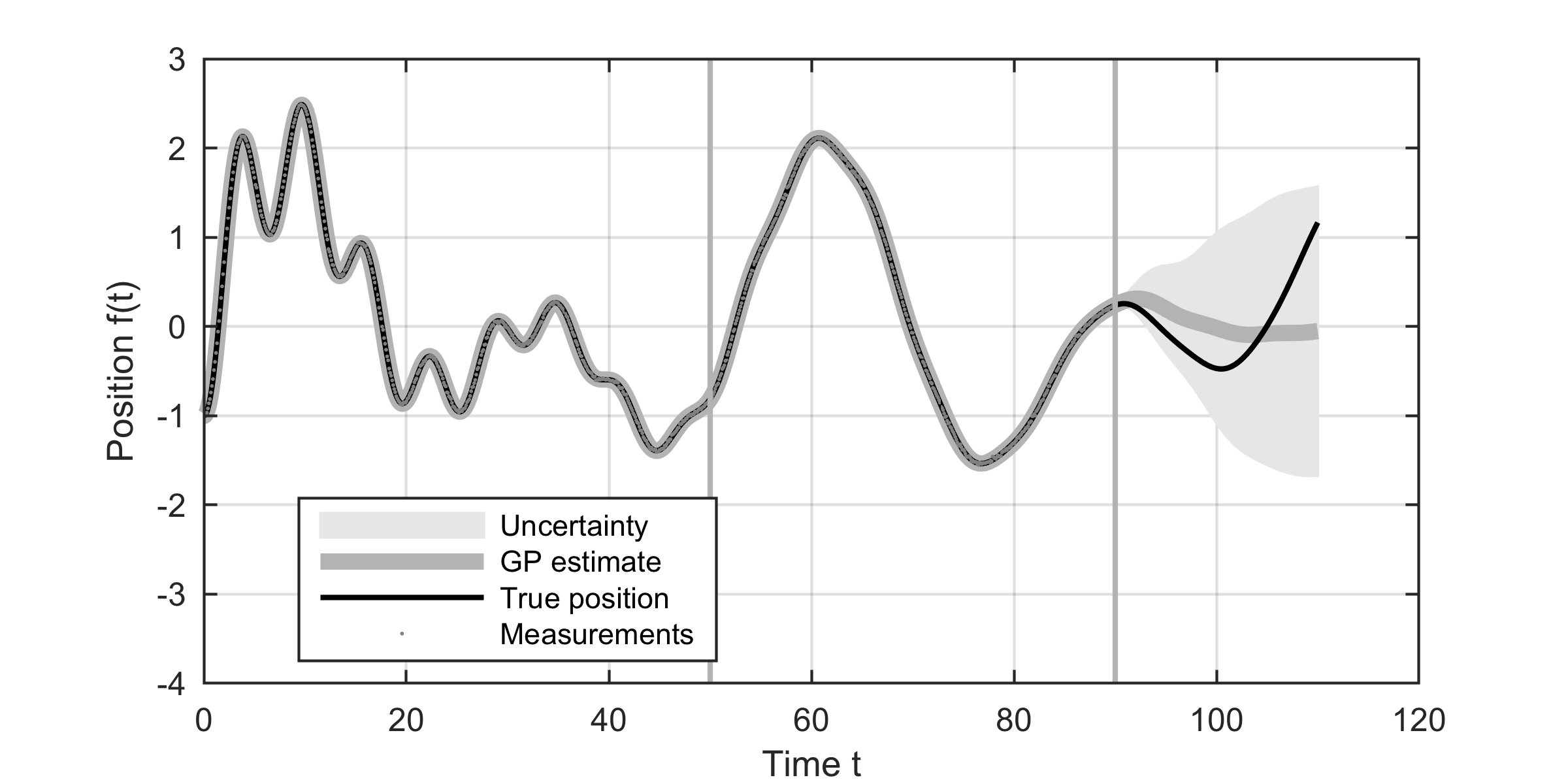}
\caption{The open-loop spring position $f(t)$ and measurements $\{y_k\}_{k=1}^n$ (which overlap with the position trajectory in the figure) along with the GP estimate and its $95\%$ uncertainty quantiles. The GP was trained using the first 50 seconds of data, after which we obtained measurements for additional 40 seconds. These time intervals are indicated with the vertical lines. 
\label{ode2:output:subfig:mlgp}}
\end{figure}

\begin{figure}[!ht]
\centering
       \includegraphics[width=0.9\columnwidth]{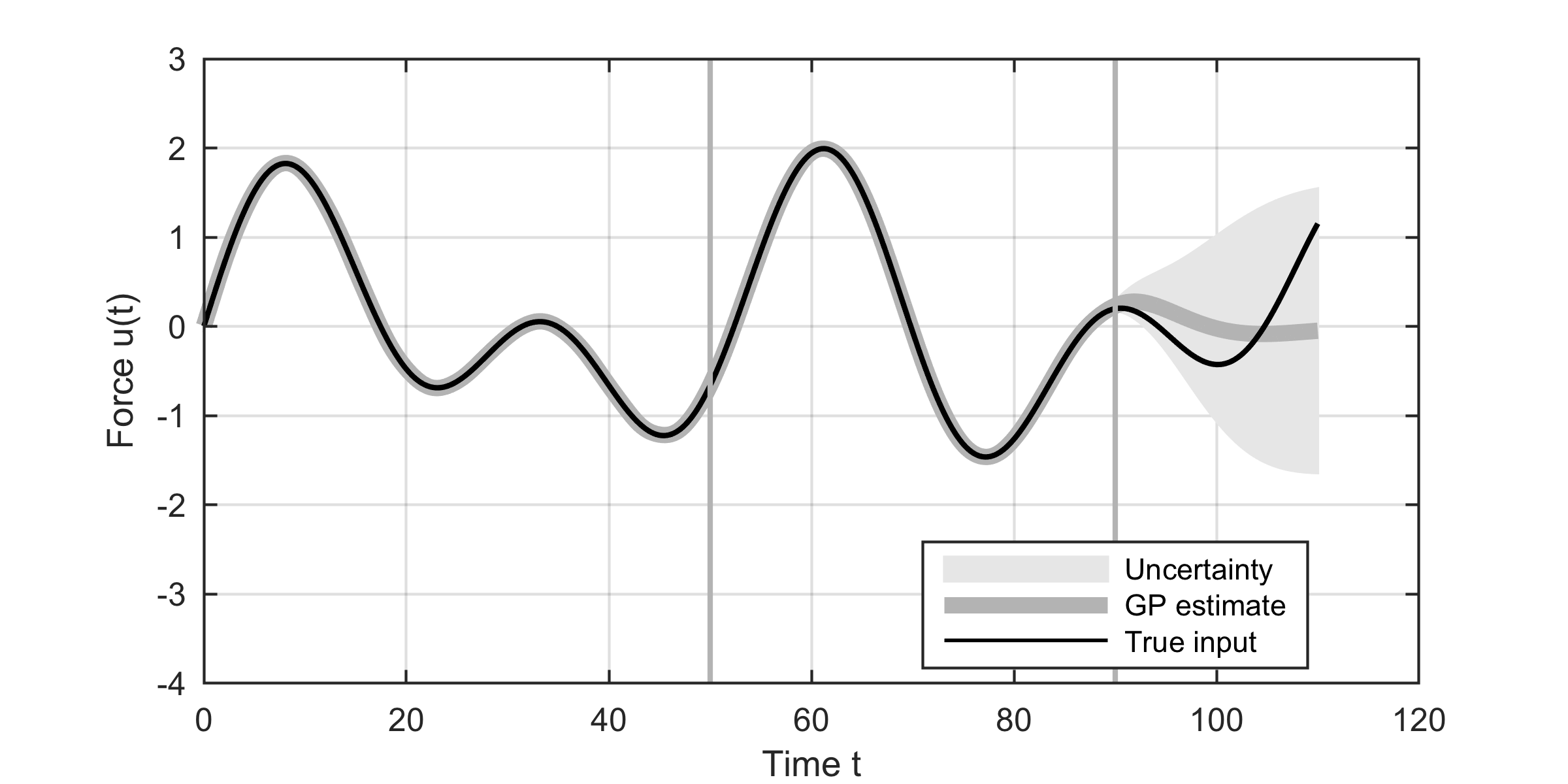}
\caption{The input signal $u(t)$ to the spring model and its GP estimate along with the $95\%$ uncertainty quantiles. \label{ode2:source:subfig:mlgp}}
\end{figure}

The result for inference for the input function $u(t)$ is shown in Figure \ref{ode2:source:subfig:mlgp}. Similarly to the position, the input estimate is good until the end of measurements after which it reverts to the zero mean. 

To demonstrate the benefit of modeling of the input signal as GP in the stochastic control context,  we consider the model \eqref{eq:sde0} with linear closed loop optimal control design for $c(t)$. Similarly to the case shown in Figures~\ref{ode2:output:subfig:mlgp} and \ref{ode2:source:subfig:mlgp}, we run the first 50 seconds without control and train the hyperparameters during this period. After that, we turn on the control signal aiming to keep the spring at zero. We consider two ways of designing the controller which were discussed in Section~\ref{sec:controlled_basic}: using the assumed separability design based on putting $u(t) = 0$ and a controller which is designed by taking into account the existence of the input signal as described in the same section. The results of using the basic linear quadratic regulator ("Basic LQR"), that is, the certainty equivalent design, and the result of using the joint LFM control ("LFM LQR") are shown in Figure~\ref{fig:cntl_spring}. It can be seen that the LFM controller is able to maintain the system much better near the origin than the basic controller. The control signals are shown in Figure~\ref{fig:cntl_signal}.

\begin{figure}[!t]
\centering
\includegraphics[width=0.9\columnwidth]{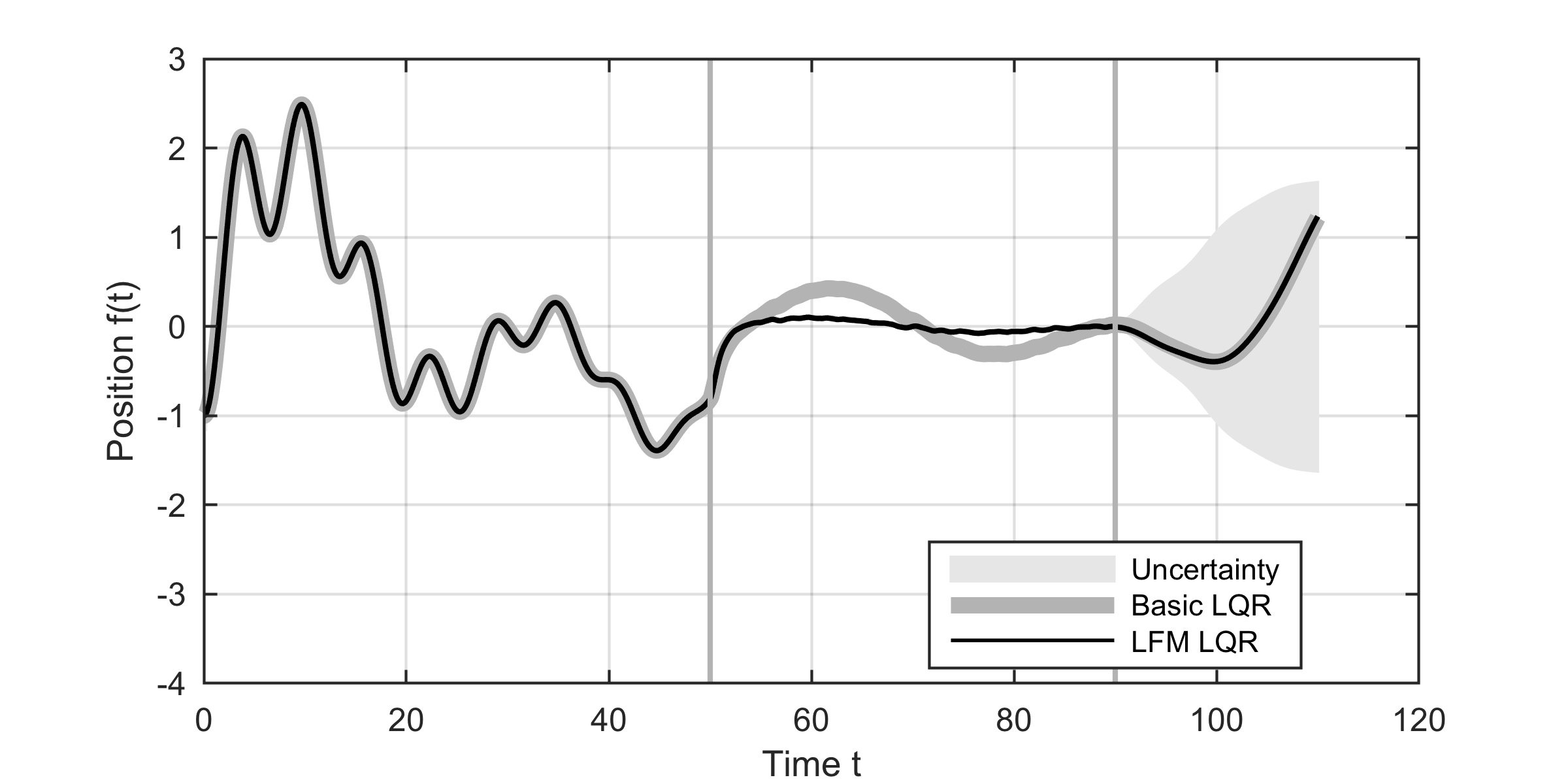}
\caption{Result of controlling the spring model with Basic LQR and LFM LQR. It can be seen the that control designed for the full LFM outperforms the basic LQR significantly. The average position tracking error for the Basic LFM was approximately $0.27$ units whereas in the case of LFM LQR it was approximately $0.11$ units.\label{fig:cntl_spring}}
\end{figure}

\begin{figure}[!t]
\centering
\includegraphics[width=0.78\columnwidth]{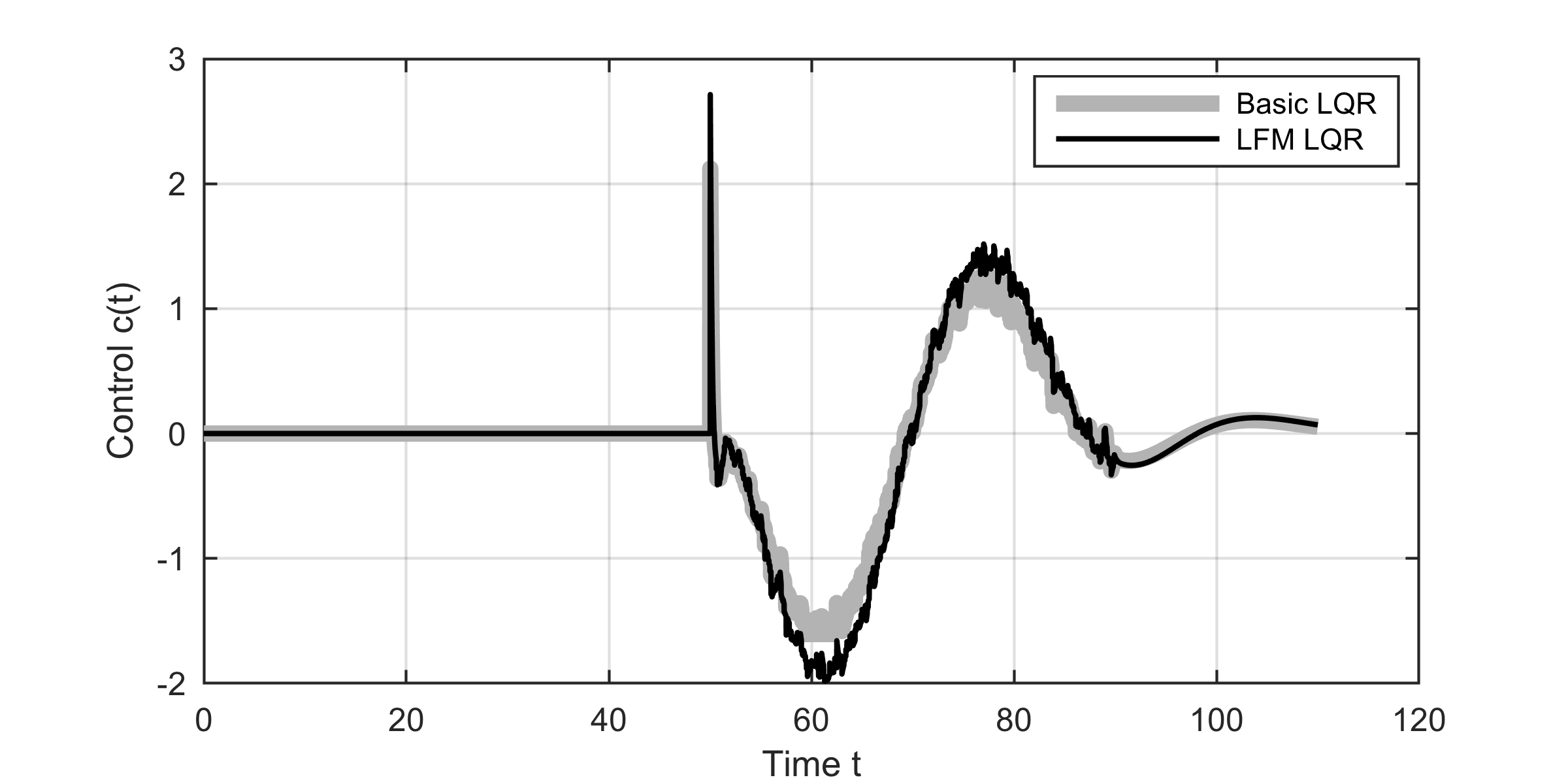}
\caption{The LQR control signals. \label{fig:cntl_signal}}
\end{figure}

\subsection{Controlled heat equation}

In this experiment we consider the controlled heat equation \eqref{eq:spde0}, where $\mathbf{x} \in \mathbb{R}^2$. Figure \ref{cartoon:heatpde} is a cartoon representation of the simulated scenario which is a heat source moving across a 2D spatial field. The field is measured at a discrete grid and the measurements are corrupted by Gaussian noise. In the simulation, the input signal $u(\mathbf{x},t)$ is the heat generated by the moving source and the aim is to reconstruct $f$ and $u$ from noisy observations as well as design an optimal control signal $c(\mathbf{x},t)$, which aims to regulate the temperature $f(\mathbf{x},t)$ to zero.

\begin{figure}[!t]
\centering
\includegraphics[width=0.45\columnwidth]{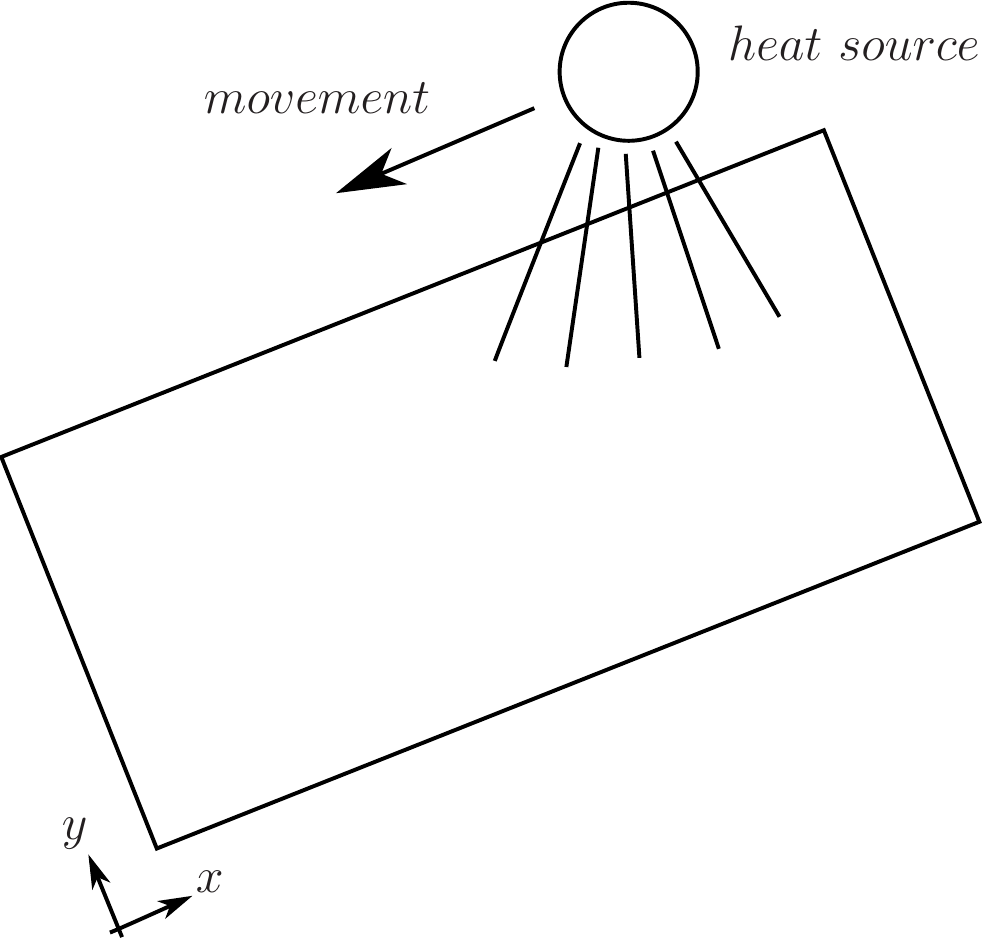}
\caption{A cartoon representation of a heat source moving across a 2D spatial field.}
\label{cartoon:heatpde}
\end{figure}

In the simulation, we used the parameters $\lambda = 0.2$ and $D = 0.001$ and the heat source was moving for 10 seconds from top-right to bottom-left direction and then it was turned off. The temperature then increases at the application point and when the heat source moves away, the position starts cooling down. Figures~\ref{fig:heat_open_x} and \ref{fig:heat_open_u} show the temperature field and the heat source at time $t = 6.9$ when no control is applied. 

\begin{figure}[!t]
\centering
   \subfloat[Temperature field $f(\mathbf{x},t)$
   \label{fig:heat_open_x}]{%
       \includegraphics[width=0.49\columnwidth]{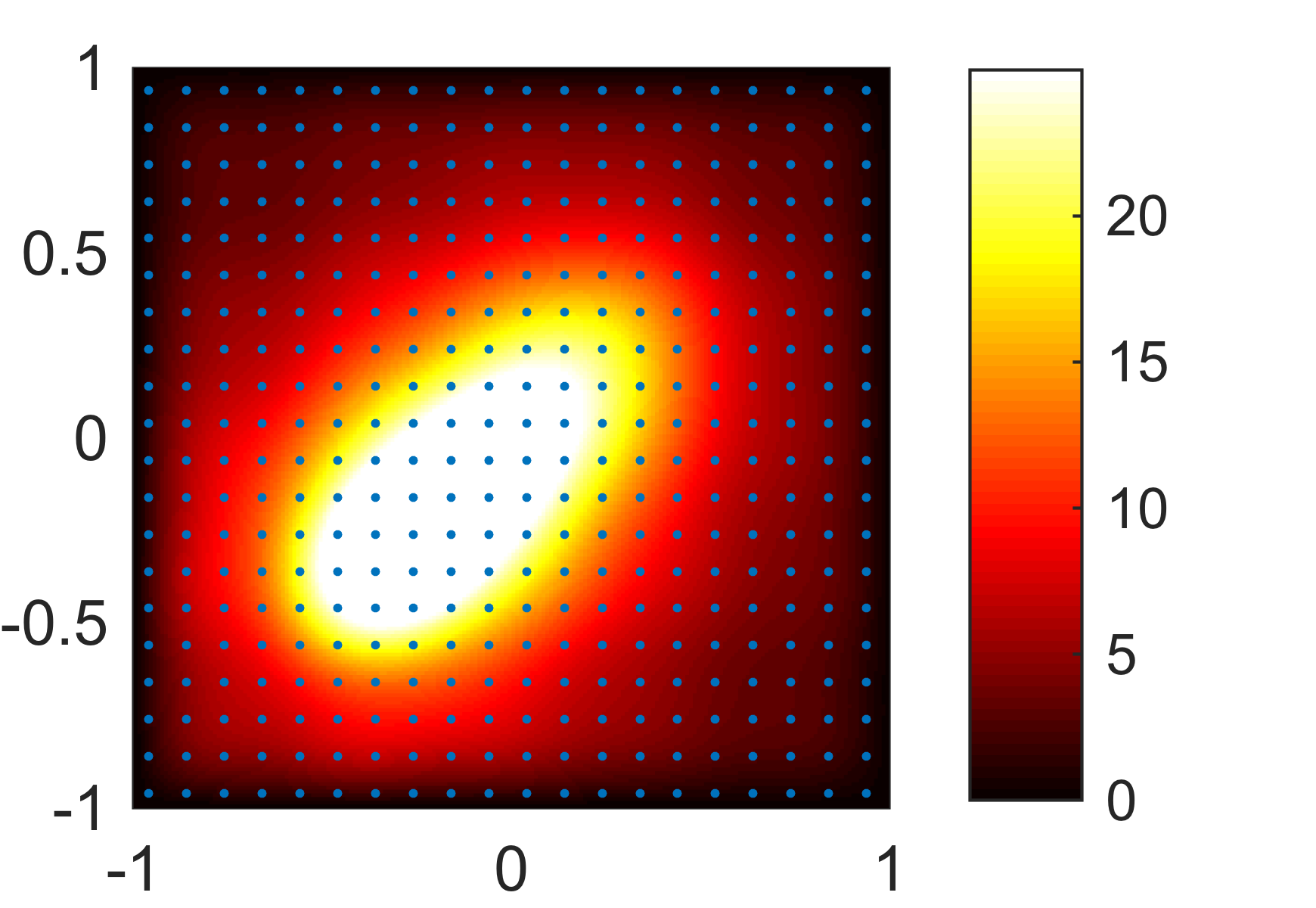}
         }
   \subfloat[Source field $u(\mathbf{x},t)$
   \label{fig:heat_open_u}]{%
       \includegraphics[width=0.49\columnwidth]{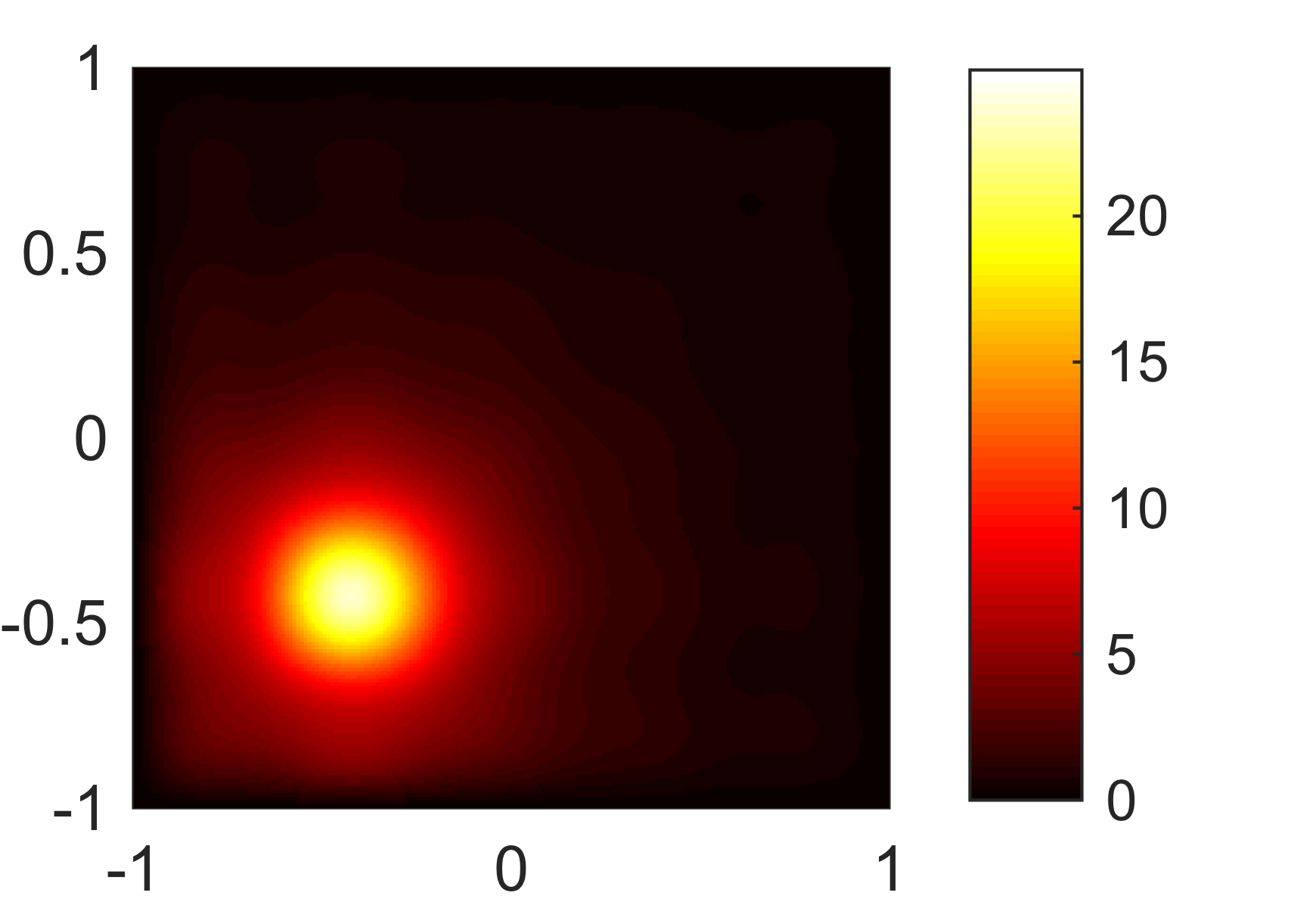}
         }
\caption{The temperature function $f(\mathbf{x},t)$ and the source function $u(\mathbf{x},t)$ at time $t = 6.9$. The small circles mark the positions of the measurements.}
\end{figure}

We then formed a Fourier-basis approximation to the PDE (with $100$ basis functions) and designed two controllers for it---one using an assumed separability design ("Basic LQR") and one by taking the input signal into account ("LFM LQR"). We used SE covariance functions for the latent force model in both time and space directions. A Kalman filter was used to estimate the physical system and input signal states from temperature measurements with low variance ($\sigma^2 = 0.01^2$) and the controller was applied using the estimate.

\begin{figure}[!t]
\centering
   \subfloat[Field $f(\mathbf{x},t)$ with Basic LQR 
   \label{fig:heat_lq_x}]{%
       \includegraphics[width=0.49\columnwidth]{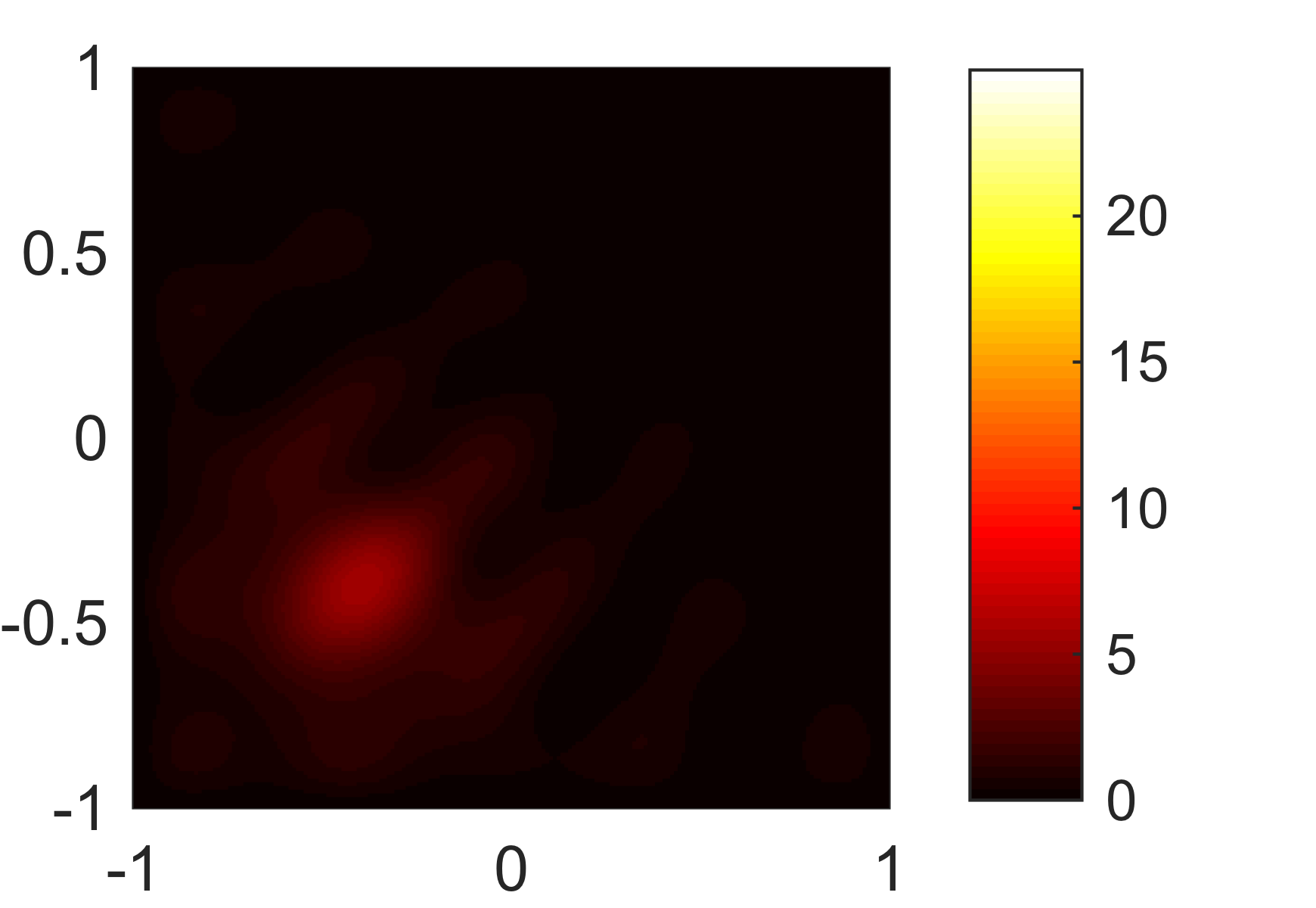}
         }
   \subfloat[Field $f(\mathbf{x},t)$ with LFM LQR
   \label{fig:heat_lfm_lq_x}]{%
       \includegraphics[width=0.49\columnwidth]{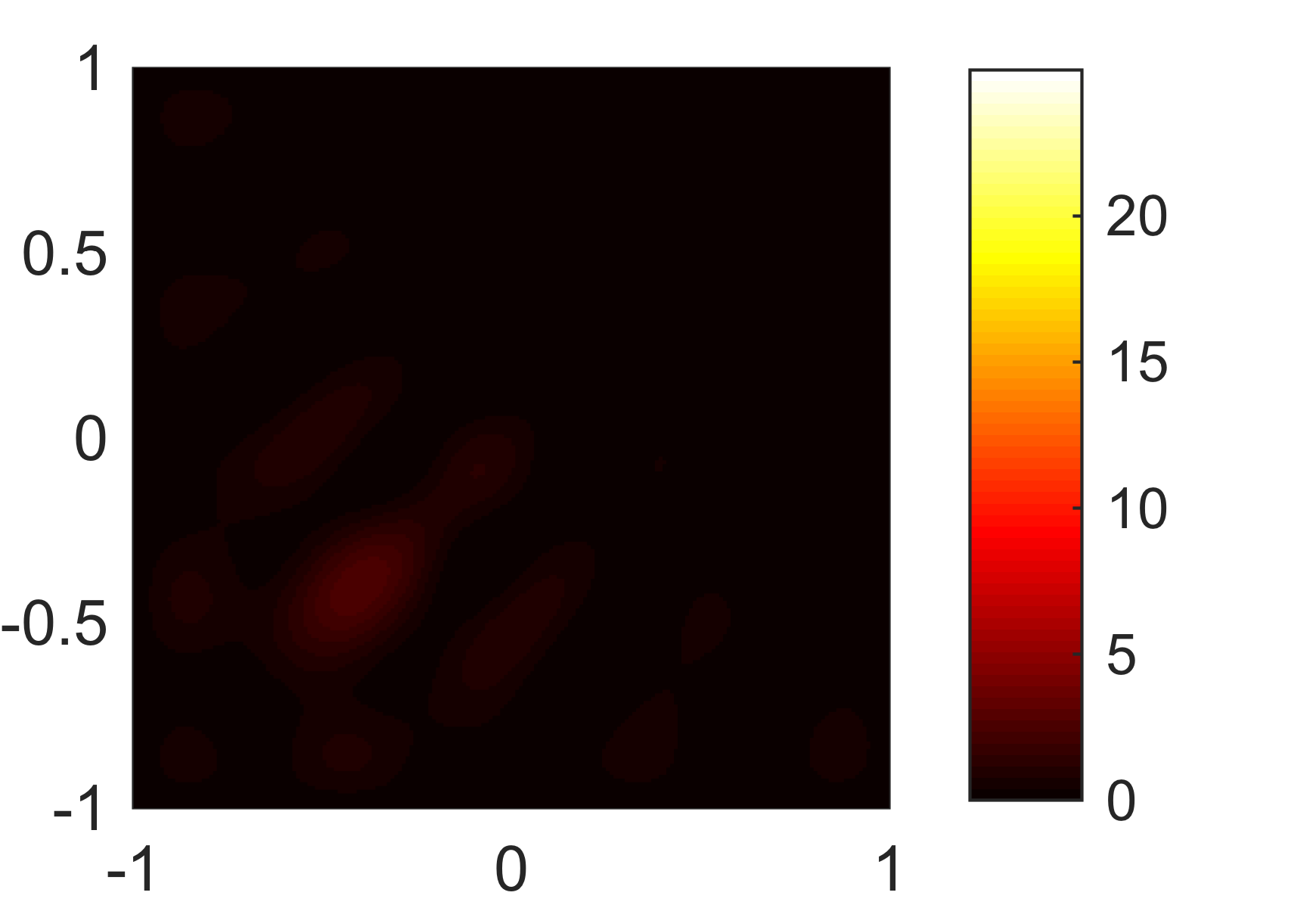}
         }
  \\
   \subfloat[Maximum temperatures
   \label{fig:heat_maxtemp}]{%
       \includegraphics[width=0.45\columnwidth]{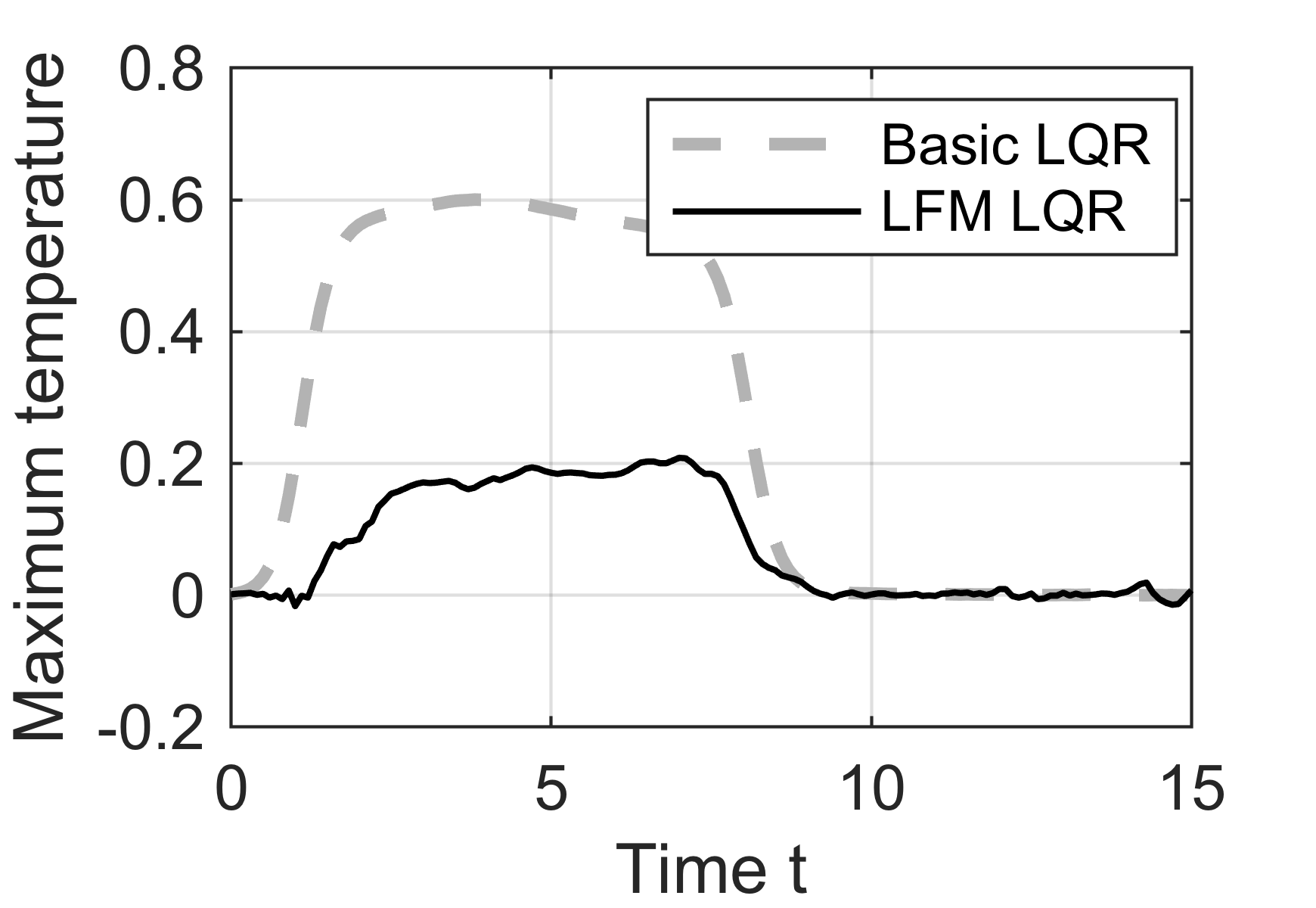}\hspace*{0.04\columnwidth}
         }
   \subfloat[LFM control signal
   \label{fig:heat_lfm_lq_c}]{%
       \includegraphics[width=0.49\columnwidth]{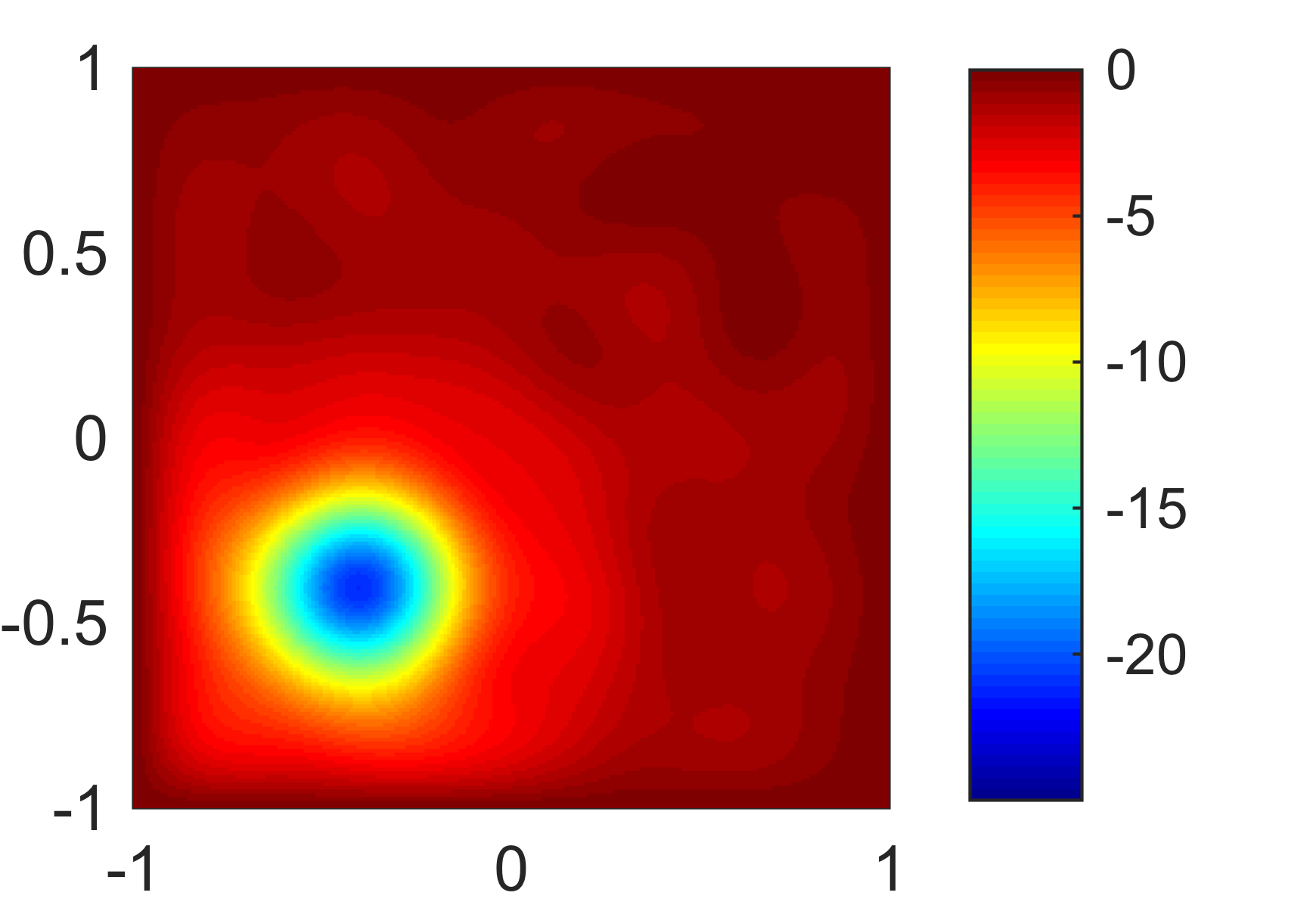}
         }
\caption{The results of using Basic LQR and LFM LQR controllers to regulate the temperature field to zero. It can be seen from Figures~\ref{fig:heat_lq_x},  \ref{fig:heat_lfm_lq_x}, and \ref{fig:heat_maxtemp} that LFM LQR is able to keep the temperature closer to zero than Basic LQR. Figure~\ref{fig:heat_lfm_lq_c} shows an example control signal which can be see to effectively cancel out the input signal part as one would expect.}
\end{figure}

Figures~\ref{fig:heat_lq_x} -- \ref{fig:heat_lfm_lq_c} show the results when the controllers were used. It can be seen that the LFM LQR provides a significantly smaller tracking error.

\section{Conclusion and Discussion}

In this paper we have studied a latent force model (LFM) framework for learning and control in hybrid models which are combinations of first-principles (physical) models and non-parametric Gaussian process (GP) models as their inputs. In particular, we have considered stochastic control problems associated with these models as well as analyzed the observability and controllability properties of the models. It turned out that although the models are often observable, they typically are not fully controllable. However, they still are output controllable with respect to the physical system part and thus the control problem is well defined. We have also experimentally shown that learning the input signal improves the control performance. This is in line with the theoretical result that the optimal control is a combination of a classical control without an input signal and an additional term that modifies the control using the knowledge on the input signal.

The framework also allows for a number of extensions. For example, introducing non-linearities in the measurement model can be tackled by replacing the Kalman filter with its non-linear counterpart (e.g., \cite{Jazwinski:1970,Maybeck:1982,Sarkka+Sarmavuori:2012,Sarkka:2013}), and another possible extension is to include an operator or a functional into the measurement model of a spatio-temporal system (e.g. \cite{Sarkka:2011,Sarkka+Hartikainen:2012,Sarkka+Solin+Hartikainen:2013}) leading to an inverse problem type of model. With these extensions the inference in the resulting system can still be performed using Kalman filter techniques and the control problem can be kept intact. In the non-linear case this corresponds to an assumed certainty equivalence approximation to the solution. It would also be possible to consider non-linear differential equation (physical) models which are driven by Gaussian processes. In that case we would need to resort to approximate Kalman filtering methods along with approximate non-linear control methods (e.g. \cite{Maybeck:1982b,Stengel:1994,Erdem:2004}).

Finally, an important practical issue is the choice of appropriate covariance function for the GP. As highlighted by the extrapolation experiment in Section~\ref{sec:ctrl_spring}, the typically used squared exponential covariance function is not always a good choice when extrapolation capability is required. The same applies to all stationary covariance functions, because they always revert to the prior mean after the data ends. One way to cope with this problem would be to use non-stationary covariance functions such as once or twice integrated stationary GPs which, instead of reverting to the prior mean, revert to zero derivative (constant prediction) or zero second derivative (linear prediction). An alternative approach would be to augment unknown constants or linear in parameters functions into the state-space model which corresponds to replacing the zero mean function with a linear in parameters model (cf.~\cite{Rasmussen+Williams:2006}). However, for these kinds of models the present observability and controllability results no longer apply as such.

\section*{Acknowledgment}

Simo S\"arkk\"a would like to thank Academy of Finland for financial support. Mauricio A. \'Alvarez has been partially financed by the EPSRC Research Project EP/N014162/1. The work was done when Neil D. Lawrence was at the University of Sheffield.

\bibliographystyle{IEEEtran}
\bibliography{reviewlfm}

\end{document}